\documentclass[a4paper,10pt]{article}
\usepackage[toc,page]{appendix}

\usepackage{ae}
\usepackage[ansinew]{inputenc}
\usepackage{lscape}
\usepackage{amssymb,amsmath,amsthm}
\usepackage{epsf}
\usepackage{psfrag}
\usepackage{mathrsfs}

\newtheorem{thm}{Theorem}
\newtheorem{lem}{Lemma}

\newtheorem{Remark}{Remark}

\theoremstyle{definition}

\theoremstyle{remark}

\newcommand{\bh}{\bar{h}}
\newcommand{\bbR}{\mathbb{R}}
\newcommand{\calC}{{\mathcal C}}
\newcommand{\calF}{{\mathcal F}}
\newcommand{\arcoth}{\operatorname{arcoth}}

\newcounter{mnotecount}[section]

\renewcommand{\themnotecount}{\thesection.\arabic{mnotecount}}

\newcommand{\mnote}[1]%{}
{\protect{\stepcounter{mnotecount}}$^{\mbox{\footnotesize $%
\!\!\!\!\!\!\,\bullet$\themnotecount}}$ \marginpar{%\color{red}
\raggedright\tiny\em $\!\!\!\!\!\!\,\bullet$\themnotecount: #1} }

\title{Spherical linear waves in de Sitter spacetime}

\author{Jo\~ao L. Costa$^{(1,3)}$, Artur Alho$^{(2)}$ and Jos\'e Nat\'ario$^{(3)}$\\\\
{\small $^{(1)}$Instituto Universitário de Lisboa (ISCTE-IUL), Lisboa, Portugal}\\
{\small $^{(2)}$Centro de Matem\'atica, Universidade do Minho, Gualtar, 4710-057 Braga, Portugal}\\
{\small $^{(3)}$Centro de An\'alise Matem\'atica, Geometria e Sistemas Din\^amicos,}\\
{\small Instituto Superior T\'ecnico, Universidade T\'ecnica de Lisboa, Portugal}
}

\begin{document}

\maketitle

\begin{abstract}
We apply Christodoulou's framework, developed to study the Einstein-scalar field equations in spherical symmetry, to the linear wave equation in de Sitter spacetime, as a first step towards the Einstein-scalar field equations with positive cosmological constant. We obtain an integro-differential evolution equation which we solve by taking initial data on a null cone. As a corollary we obtain elementary derivations of expected properties of linear waves in de Sitter spacetime: boundedness in terms of (characteristic) initial data, and a Price law establishing uniform exponential decay, in Bondi time, to a constant.

% With the Einstein-scalar field equations with positive cosmological constant in mind, we apply Christodoulou's framework, developed to study the vanishing cosmological constant case, to spherically symmetric solutions of the linear wave equation in de Sitter spacetime. We obtain an integro-differential evolution equation which we solve by taking initial data on a null cone. As a corollary we obtain elementary derivations of expected properties of linear waves in de Sitter spacetime: boundedness in terms of (characteristic) initial data, and a Price law establishing uniform exponential decay, in Bondi time, to a constant.
\end{abstract}

\section{Introduction}

The study of the linear wave equation
\begin{equation}
\label{onda}
\square_g\phi=0
\end{equation}
on fixed backgrounds $(M,g)$ is a stepping stone to the analysis of the nonlinearities of gravitation.
%, from cosmic censorship to non-linear stability~\cite{DRlectures}. 
In this paper we apply Christodoulou's framework, developed in~\cite{CommMathPhys1986}, to spherically symmetric solutions of~\eqref{onda} on a de Sitter background, as a prerequisite to the study of the coupled Einstein-scalar field equations with positive cosmological constant in spherical symmetry (which will be pursued elsewhere). If the cosmological constant vanishes then the uncoupled problem~\eqref{onda} is trivial\footnote{This can be seen from the fact that operator $\cal{F}$ in equation~\eqref{defIntF}, whose fixed points are the solutions of~\eqref{onda}, is a constant operator for $\Lambda = 0$; when perturbing to the nonlinear problem this operator becomes a contraction for small initial data.}, a fact that Christodoulou explored in~\cite{CommMathPhys1986} to solve the coupled case for suitably small initial data. For positive cosmological constant, however, the uncoupled case is more complicated\footnote{In this case the operator $\cal{F}$ is not even a contraction in the full domain.}, and it is essential to understand it thoroughly in order to ascertain how much freedom is there when perturbing it to the nonlinear case, as well as to determine which decays to expect and which function spaces to use.

Following Christodoulou, we turn~\eqref{onda} into an integro-differential evolution equation, which we solve by taking initial data on a null cone. This step, which is trivial in the case of vanishing cosmological constant, turns out to be quite subtle for positive cosmological constant. As a corollary we obtain elementary derivations of expected properties of linear waves in de Sitter spacetime: boundedness in terms of (characteristic) initial data, and uniform exponential decay, in Bondi time, to a constant (from which exponential decay to a constant in the usual static time coordinate easily follows; such boundedness and decay results may be seen, respectively, as analogues of the Kay-Wald theorem~\cite{KayWald} and of a Price law~\cite{Price, Dafermos}, both originally formulated for Cauchy data in a Schwarzschild background). Although widely expected from the stability of de Sitter spacetime \cite{Friedrich}, we are unaware of a written proof of uniform exponential decay. Similarly, the bound that we obtain for the solution in terms of the $C^0$ norm of the (characteristic) initial data is, to the best of our knowledge, original (notice that in particular the bound \eqref{hBound} involves no ``loss of derivatives"). Another novel aspect of our work is the fact that our results apply to a domain containing the cosmological horizon in its interior, therefore including both a local and a cosmological regions (as opposed to considering only the local region). This allows us to determine decays for initial data which lead to uniform exponential decay in time of the solution.
%Finally, the bound of solutions we obtain is in terms of the $C^0$ norm of the (characteristic) initial data and therefore has no ``loss of derivatives"; again, although expected, since the ``loss of derivatives" is usually connected to the existence of a photon sphere, at $r=3M$ in Schwarzschild-de-Sitter~\cite{DRlectures}, we are unaware of a published result along this lines.~\jlca{Isn't this lack of loss of derivatives consequence of the data being characteristic instead of Cauchy}.

Numerical evidence for pointwise exponential decay can be found in~\cite{Brady}, and references therein, where higher spherical harmonics are also studied, as well as the non-linear system. Yagdjian and Galstian~\cite{CommMathPhys2009} constructed the fundamental solutions of~\eqref{onda} in de Sitter spacetime and proved exponential decay of certain homogeneous Sobolev $L^p$ norms, $2 \leq p<\infty$. Also along these lines, Ringstr\"om \cite{Ringstrom} obtained exponential decay for non-linear perturbations of locally de Sitter cosmological models in the context of the Einstein-nonlinear scalar field system with a positive potential. Finally, exponential pointwise decay in the local region between the black hole and the cosmological horizons in a Schwarzschild-de-Sitter spacetime follows from the papers by Dafermos and Rodnianski~\cite{arxiv0709.2766, DRlectures} (see also~\cite{Hafner}).

% By comparison, the results presented here suffer from the requirement of symmetry. In fact, the methods used rely extensively on the assumption of spherical symmetry of solutions and on the existence of a regular center of symmetry, and it is not clear if it is possible to extend them to the study of higher spherical harmonics, or to the analysis of linear waves in other backgrounds, like Schwarzschild-de-Sitter\footnote{For a thorough discussion of linear waves in black hole spacetimes see~\cite{DRlectures}.}.
%\footnote{One should also note the following groundbreaking results concerning the the local region of Schwarzschild-de-Sitter with non-vanishing mass:  in~\cite{arxiv0709.2766} boundedness in terms of Cauchy data is established as well as exponential decay, in double null coordinates, for all spherical harmonics, with all results holding up to the horizons; in~\cite{}, under stronger conditions on initial data, exponential decay of $\phi$ is established. For a thorough discussion on linear waves on black holes see~\cite{DRlectures}.}.

\section{Christodoulou's framework for spherical waves}
\label{sectionBondi}

Bondi coordinates~\cite{CommMathPhys1986} $(u,r,\theta,\varphi)$ map the causal future of any point in de Sitter spacetime isometrically onto
$\left([0,\infty)\times[0,\infty)\times S^2,g\right)$, where
\begin{equation}
 g=-\left(1-\frac{\Lambda}{3}r^{2}\right)du^{2}-2dudr+r^{2}d\Omega^{2}\;,
\label{dS_Bondi}
\end{equation}
with $d\Omega^{2}$ the round metric of the two-sphere (cf.~Figure~\ref{Penrose}).

\begin{figure}[h!]
\begin{center}
\psfrag{r=0}{$r=0$}
\psfrag{i}{$i$}
\psfrag{H}{$\mathcal{H}$}
\psfrag{I+}{$\mathscr{I^+}$}
\epsfxsize=.4\textwidth
\leavevmode
\epsfbox{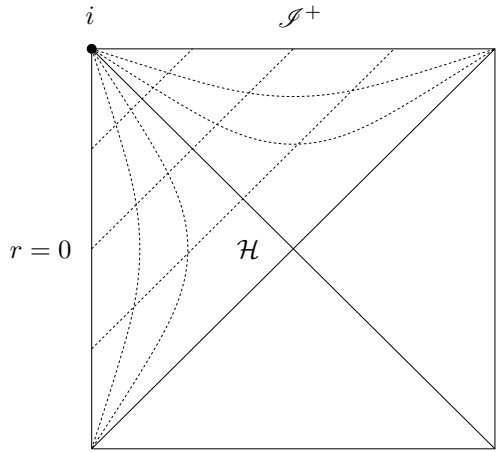}
\end{center}
\caption{Penrose diagram of de Sitter spacetime. The lines $u=\text{constant}$ are the outgoing null geodesics starting at $r=0$. The point $i$ corresponds to $u=+\infty$, the cosmological horizon $\mathcal{H}$ to $r=\sqrt{\frac3\Lambda}$ and the future null infinity $\mathscr{I^+}$ to $r=\infty$.} \label{Penrose}
\end{figure}

In these coordinates the wave equation
\begin{equation*}
\square_g\phi=0\Leftrightarrow\partial_{\mu}\left(\sqrt{-\det (g)}\,\partial^{\mu}\phi\right)=0\;,
\end{equation*}
for spherically symmetric functions, $\partial_\theta\phi=\partial_\varphi\phi=0$, reads
\begin{equation}
\label{wave_equation_dS}
-2r\frac{\partial}{\partial r}\left(\frac{\partial\phi}{\partial u}\right)-2\frac{\partial\phi}{\partial u}+r\left(1-\frac{\Lambda}{3}r^{2}\right)\frac{\partial^{2}\phi}{\partial r^{2}}+\left(2-\frac{4}{3}\Lambda r^{2}\right)\frac{\partial\phi}{\partial r}=0\;.
\end{equation}
Following Christodoulou~\cite{CommMathPhys1986} we consider the change of variable
\begin{equation*}
    \label{hdef}
    h:=\frac{\partial}{\partial r}\left(r\phi\right)\;.
\end{equation*}
If we assume that
$$\lim_{r\rightarrow 0} r\phi=0\;,$$
it immediately follows that
\begin{equation}
\phi=\bar{h}:=\frac{1}{r}\int^{r}_{0}h\left(u,s\right)ds\quad\text{ and }\quad\frac{\partial\phi}{\partial r}=\frac{\partial\bar{h}}{\partial r}=\frac{h-\bar{h}}{r}\;.
\label{Definitionh}
\end{equation}
Moreover, assuming that the crossed partial derivatives of $r\phi$ commute, we see that~\eqref{wave_equation_dS} is equivalent to
\begin{equation}
\label{mainEq}
    Dh=-\frac{\Lambda}{3}r(h-\bar{h})\;,
\end{equation}
where $D$ is the differential operator given by
\begin{equation*}
 D:=\frac{\partial}{\partial u}-\frac{1}{2}\left(1-\frac{\Lambda}{3}r^{2}\right)\frac{\partial}{\partial r}\;.
\end{equation*}

\section{Main result: statement and proof}

 Our main result is the following

\begin{thm}
\label{mainThm}
Let $\Lambda> 0$. Given $h_0\in\mathcal{C}^{k}([0,\infty))$, for some $k\geq1$, the problem
\begin{equation}
\label{mainEqChar}
\left\{
\begin{array}{l}
  Dh = -\frac{\Lambda}{3}r(h-\bar{h}) \\
  h(0,r) = h_0(r)
\end{array}
\right.
\end{equation}
has a unique solution  $h\in\mathcal{C}^{k}([0,\infty)\times[0,\infty))$.

Moreover, if $\|h_0\|_{\mathcal{C}^0}$ is finite\footnote{Recall that if $f: X \to \bbR$ is continuous and bounded then $\|f\|_{\mathcal{C}^0}=\sup_{x \in X} |f(x)|$.} then
\begin{equation}\label{hBound}
   \|h\|_{\mathcal{C}^0}=\|h_0\|_{\mathcal{C}^0}\;.
\end{equation}

Also, if $\|(1+r)^p \partial_rh_0\|_{\mathcal{C}^0}$ is finite for some $0\leq p\leq 4$ and $H\leq2\sqrt{\frac{\Lambda}{3}}$ then
\begin{equation}\label{dhBound}
   \|(1+r)^pe^{Hu}\partial_rh\|_{\mathcal{C}^0}\lesssim  \|(1+r)^p \partial_rh_0\|_{\mathcal{C}^0}\;,
\end{equation}
and, consequently, there exists $\underline{h}\in\mathbb{R}$ such that
%
%\begin{equation}\label{expDecay}
%    |h(u,r)-\underline{h}|\lesssim r^{3-p}e^{-Hu} \quad \text{, for }\quad 3\leq p \leq 4 \;,
%\end{equation}
%%
%\begin{equation}\label{expDecay2}
%    |h(u,r)-\underline{h}|\leq C(r)e^{-Hu} \quad \text{, for } \quad 0\leq p < 3\;.
%\end{equation}
%%
\begin{equation}\label{expDecay}
    |h(u,r)-\underline{h}|\lesssim(1+r)^{n(p)}e^{-Hu}\;,
\end{equation}
with
\begin{equation}\label{lp}
n(p)=\left\{
\begin{array}{ccc}
%-1  & , & p>4 \\
0 & , & 2 < p\leq 4 \\
2   & , & 0\leq p \leq 2
\end{array}\right.\;.
\end{equation}
\end{thm}

\begin{Remark}
\label{remMainThm}
The powers of $1+r$ obtained are far from optimal. Since we are mainly interested in understanding whether the decay in $u$ obtained by this method is uniform in $r$, we were only careful in computing precise estimates for $2 < p \leq 4$, which is enough to establish uniform decay for $p>2$ (if $p>4$ the $p=4$ result applies, and in fact it does not seem to be possible to obtain a stronger decay in $r$ for $\partial_r h$). For $p\leq 2$ our method does not provide uniform decay, but it is not clear if this is an artifact of these techniques or an intrinsic property of spherical linear waves in de Sitter.
\end{Remark}

\begin{proof}

For $h\in\calC^0([0,\infty)\times[0,\infty))$, we have $r\bh\in\calC^0([0,\infty)\times[0,\infty))$, and so we can define $\calF(h)$ to be the solution to the linear equation
\begin{equation}
\label{defF}
\left\{
\begin{array}{l}
  D(\calF(h)) = -\frac{\Lambda}{3}r(\calF(h)-\bar{h}) \\
  \calF(h)(0,r) = h_0(r)
\end{array}
\right.\;.
\end{equation}
The integral lines of $D$ (incoming light rays in de Sitter), which satisfy
\begin{equation}\label{char}
    \frac{dr}{du}=-\frac{1}{2}\left(1-\frac{\Lambda}{3}r^{2}\right)\;,
\end{equation}
are characteristics of the problem at hand. Integrating~\eqref{defF} along such characteristics we obtain
\begin{equation}
\label{defIntF}
 \calF\left(h\right)(u_1,r_1)=h_0(r(0))e^{-\frac{\Lambda}{3}\int_0^{u_1} r(s)ds}
 +\frac{\Lambda}{3}\int_0^{u_1} r(v)\bh (v,r(v)) e^{-\frac{\Lambda}{3}\int_v^{u_1} r(s)ds} dv\;,
\end{equation}
where, to simplify the notation, we denote the solution to~\eqref{char} satisfying $r(u_1)=r_1$ simply by $s\mapsto r(s)$; we are dropping any explicit reference to the dependence on $(u_1,r_1)$, but it should be noted, in particular, that $r(0)$ is an analytic function of $(u_1,r_1)$.

%%%%%%%%%%%%%%%%%%%%%%%%%%%%%%
\newcommand{\our}{\calC^0_{U,R}}
\newcommand{\meuL}{\frac{\Lambda}{3}}
%%%%%%%%%%%%%%%%%%%%%%%%%%%%%%%%

Given $U,R>0$, let $\our$ denote the Banach space $\left(\calC^0\left([0,U]\times[0,R]\right), \|\cdot\|_{\our}\right)$, where
\begin{equation}\label{defOur}
    \|f\|_{\our}=\sup_{(u,r)\in[0,U]\times[0,R]}|f(u,r)|\;.
\end{equation}

Let $r_c:=\sqrt{\frac{3}{\Lambda}}$ be the unique non-negative zero of $1-\frac{\Lambda}{3}r^{2}$ (see~\eqref{dS_Bondi} and~\eqref{char}). The non-decreasing behavior of the characteristics satisfying $r_1\geq r_c$ shows that the restriction of $\calF$ to $\our$ is well defined for all $R\geq r_c$. In fact:

\begin{lem}
Given $U>0$ and $R\geq r_c:=\sqrt{\frac{3}{\Lambda}}$, $\calF$ contracts in $\our$.
\end{lem}

\begin{proof}

 Fix $U>0$ and $R\geq r_c$. Then

\begin{equation*}
\begin{aligned}
 \|\calF(h_1)-\calF(h_2)\|_{\our}
 &=
 \sup_{(u_1,r_1)\in[0,U]\times[0,R]}\left|\calF(h_1)(u_1,r_1)-\calF(h_2)(u_1,r_1)\right|
 \\
 &\leq
 \sup_{(u_1,r_1)\in[0,U]\times[0,R]}
 \left\{\meuL\int_0^{u_1} r(v)\left|\bh_1(v,r(v))-\bh_2(v,r(v))\right|e^{-\frac{\Lambda}{3}\int_v^{u_1} r(s)ds} dv\right\}
 \\
 &\leq
 \sup_{(u_1,r_1)\in[0,U]\times[0,R]} \left\{\int_0^{u_1} \meuL r(v)e^{-\frac{\Lambda}{3}\int_v^{u_1} r(s)ds} dv\right\}
  \cdot \|\bh_1-\bh_2\|_{\our}
\\
 &\leq
 \sup_{(u_1,r_1)\in[0,U]\times[0,R]} \left\{\left[e^{-\frac{\Lambda}{3}\int_v^{u_1} r(s)ds}\right]_{v=0}^{u_1}\right\}
 \cdot \sup_{(u,r)\in[0,U]\times[0,R]}\left\{\frac{1}{r}\int_0^r|h_1(u,s)-h_2(u,s)|ds\right\}
\\
 &\leq
 \underbrace{\sup_{(u_1,r_1)\in[0,U]\times[0,R]} \left\{1-e^{-\frac{\Lambda}{3}\int_0^{u_1} r(s)ds}\right\}}_{:=\sigma}
 \cdot \|h_1-h_2\|_{\our}\;.
\end{aligned}
\end{equation*}

 Throughout, to obtain estimates, and in particular to estimate $\sigma$, one needs to consider three (causally) separate regions, naturally corresponding to the bifurcations of~\eqref{char}: the local region ($r<r_c$), the cosmological horizon ($r=r_c$), and the cosmological region ($r>r_c$). However, since the computations are similar we will only present the details concerning the most delicate case, $r>r_c$\;.

The solution to~\eqref{char} satisfying $r_1=r(u_1)>r_c:=\sqrt{\frac{3}{\Lambda}}$, is given by
\begin{equation}
\label{charCosmic}
r(u)=\sqrt{\frac{3}{\Lambda}}\coth{\left(\frac{1}{2}\sqrt{\frac{\Lambda}{3}}(c-u)\right)}\;,
\end{equation}
where 
\[
c=u_1 + 2\sqrt{\frac{3}{\Lambda}}\arcoth\left(\sqrt{\frac{\Lambda}{3}} \, r_1 \right)
\]
(in particular $c>u_1$, and so \eqref{charCosmic} is well defined for $0 \leq u \leq u_1$). It follows that
\begin{equation*}
\begin{aligned}
 -\meuL\int^{u_1}_{0}r(s)ds
 &=
 \int^{u_1}_{0}-\sqrt{\frac{\Lambda}{3}}\coth{\left(\frac{1}{2}\sqrt{\frac{\Lambda}{3}}(c-s)\right)}ds
 \\
 &=
 \int^{u_1}_{0}2\frac{d}{ds}\ln{\left[\sinh{\left(\frac{1}{2}\sqrt{\frac{\Lambda}{3}}(c-s)\right)}\right]}ds
%\\
=
 \ln{\left[\frac{\sinh{\left(\frac{1}{2}\sqrt{\frac{\Lambda}{3}}(c-u_1)\right)}}
 {\sinh{\left(\frac{1}{2}\sqrt{\frac{\Lambda}{3}}c\right)}}\right]^2}\;,
\end{aligned}
\end{equation*}
and consequently
\begin{equation*}
\begin{aligned}
e^{-\meuL\int^{u_{1}}_{0}r(s)ds}&=\frac{\sinh^2{\left(\frac{1}{2}\sqrt{\frac{\Lambda}{3}}(c-u_1)\right)}}{\sinh^2{\left(\frac{1}{2}\sqrt{\frac{\Lambda}{3}}c\right)}}
                                 =\frac{\cosh^2{\left(\frac{1}{2}\sqrt{\frac{\Lambda}{3}}(c-u_1)\right)}}{\sinh^2{\left(\frac{1}{2}\sqrt{\frac{\Lambda}{3}}c\right)}\coth^2{\left(\frac{1}{2}\sqrt{\frac{\Lambda}{3}}(c-u_1)\right)}} \\
                                &=\left[\frac{\cosh{\left(\alpha(c-u_1)\right)}}{\sinh{\left(\alpha c\right)}}\frac{1}{2\alpha r_1}\right]^2
                                 =\left[\frac{e^{\alpha(c-u_1)}+e^{-\alpha(c-u_1)}}{e^{\alpha c}-e^{-\alpha c}}\right]^{2}\frac{1}{4\alpha^2 r^{2}_1}  \\
                                &\geq\frac{e^{-2\alpha u_1}}{4\alpha^2 r^2_1},
\end{aligned}
\end{equation*}
where $\alpha := \frac{1}{2}\sqrt{\frac{\Lambda}{3}}$. Define
\begin{equation*}
\begin{aligned}
 \sigma_{cosm}\left(U,R\right)
 &:=
 \sup_{(u_1,r_1)\in[0,U]\times(r_c,R]}\left(1-e^{-\meuL\int^{u_{1}}_{0}r(s)ds}\right)
 \\
 &\leq
 \sup_{(u_1,r_1)\in[0,U]\times(r_c,R]}\left(1-\frac{e^{-2\alpha u_1}}{4\alpha^2 r^2_1}\right)\leq\left(1-\frac{3}{\Lambda}\frac{e^{-\sqrt{\frac{\Lambda}{3}}U}}{R^2}\right)<1\;.
\end{aligned}
\end{equation*}
Similar computations give
$$\sigma_{loc}:=
 \sup_{(u_1,r_1)\in[0,U]\times[0,r_c)}\left(1-e^{-\meuL\int^{u_{1}}_{0}r(s)ds}\right)\leq\left(1-\frac{e^{-\sqrt{\frac{\Lambda}{3}}U}}{4}\right)<1\;,$$
for the local region,
and
$$\sigma_{hor}:=\sup_{u_1\in[0,U]}\left(1-e^{-\meuL\int^{u_{1}}_{0}r_cds}\right)
\leq 1-e^{-\sqrt{\frac{\Lambda}{3}}U}<1\;,$$
along the cosmological horizon. Finally $\sigma=\max\{\sigma_{loc},\sigma_{hor},\sigma_{cosm}\}<1$, and the statement of the lemma follows.

\end{proof}

By the contraction mapping theorem \cite{GT}, given $U>0$ and $R\geq r_c$, there exits a unique fixed point $h_{U,R}\in\our$ of $\calF$. Uniqueness guarantees that in the intersection of two rectangles $[0,U_1]\times[0,R_1]\cap[0,U_2]\times[0,R_2]$ the corresponding $h_{U_1,R_1}$ and $h_{U_2,R_2}$ coincide. Consequently, there exists a unique continuous map $h:[0,\infty)\times[0,\infty) \to \bbR$ such that $h=\calF(h)$, i.e.,
\begin{equation}
\label{fixedPoint}
 h(u_1,r_1)=h_0(r(0))e^{-\frac{\Lambda}{3}\int_0^{u_1} r(s)ds}
 +\frac{\Lambda}{3}\int_0^{u_1} r(v)\bh (v,r(v)) e^{-\frac{\Lambda}{3}\int_v^{u_1} r(s)ds} dv\;,
\end{equation}
in $[0,\infty)\times[0,\infty)$. Continuity of $h$ implies continuity of $r\bh$, so we are allowed to differentiate~\eqref{fixedPoint} in the direction of $D$, which proves that $h$ is in fact a ($\calC^0$) solution of~\eqref{mainEqChar}. Existence and uniqueness in $\calC^0\left([0,\infty)\times[0,\infty)\right)$ follow.

\vspace{0,2cm}

To see that a solution of~\eqref{mainEqChar} is as regular as its initial condition assume that $h_0\in\calC^{k+1}$, $k\geq 0$, and start by noticing that if $h\in\calC^k$ then $r\bh$ and $\partial_r(r\bh)$ are also in $\calC^k$. In particular for $h\in\calC^0$ we can differentiate~\eqref{fixedPoint} with respect to $u_1$ to obtain
\begin{equation}\label{dhdu}
\begin{aligned}
\frac{\partial h}{\partial u_1}
&=
\frac{\partial}{\partial u_1}\left(h_0(r(0))e^{-\frac{\Lambda}{3}\int_0^{u_1} r(s)ds}\right)+\meuL\left(r\bh\right)(u_1,r_1)
\\
&+
\meuL \int_0^{u_1} \frac{\partial (r\bh)}{\partial r} (v,r(v)) \frac{\partial r}{\partial u_1} (v) e^{-\frac{\Lambda}{3}\int_v^{u_1} r(s)ds} dv
\\
&+
\frac{\Lambda}{3}\int_0^{u_1} r(v)\bh (v,r(v)) \frac{\partial}{\partial u_1}\left(e^{-\frac{\Lambda}{3}\int_v^{u_1} r(s)ds}\right)\;.
\end{aligned}
\end{equation}
This last expression shows that $h_0\in\calC^{k+1}$ and $h\in\calC^k$ implies $\frac{\partial h}{\partial u_1}\in\calC^{k}$. The same reasoning works for the derivative with respect to $r_1$ (note that although $r(v)$ is given by different expressions according to whether $r_1<r_c$, $r_1=r_c$ or $r_1>r_c$, it is still the solution of a smooth ODE satisfying $r(u_1)=r_1$, and as such depends smoothly on the data $(u_1,r_1)$). Consequently,  if $h_0\in C^{k+1}$ and $h\in\calC^k$, then $h$ is in fact in $C^{k+1}$ and the regularity statement follows by induction.

\vspace{0,2cm}

To establish~\eqref{hBound} first note that:

\begin{lem}
\label{calFdecrease}
If $\|h_0\|_{\calC^0}\leq y_0$ and $\|h\|_{\calC^0}\leq y_0$, for some $y_0\geq 0$, then $\|\calF(h)\|_{\calC^0}\leq y_0$.
\end{lem}

\begin{proof}
From~\eqref{defIntF} we see that
$$
\begin{aligned}
|\calF(h)(u_1,r_1)|
&\leq
\|h_0\|_{\calC^0}\;e^{-\frac{\Lambda}{3}\int_0^{u_1} r(s)ds}
+\|\bh\|_{\calC^0}\;\meuL \int_0^{u_1} r(v) e^{-\frac{\Lambda}{3}\int_v^{u_1} r(s)ds} dv
\\
&\leq y_0 \underbrace{\left(e^{-\frac{\Lambda}{3}\int_0^{u_1} r(s)ds}
+\meuL \int_0^{u_1} r(v) e^{-\frac{\Lambda}{3}\int_v^{u_1} r(s)ds} dv\right)}_{\equiv1} =y_0\;.
\end{aligned}
$$
The last step follows by a direct computation, as before, or by noticing that since $h\equiv1$ is a solution to~\eqref{mainEqChar}, with $h_0\equiv1$, one has $\calF(1)\equiv1$.

\end{proof}

Now consider the sequence
$$
\left\{
\begin{array}{l}
h_0(u,r)=h_0(r) \\
h_{n+1}=\calF(h_n)
\end{array}
\right.\;.
$$
We have already established that, for any $U>0$ and $R\geq r_c$, $h_n$ converges in $\our$ to $h$, the solution of~\eqref{mainEqChar}. Lemma~\ref{calFdecrease} then tells us that
$$\|h_n\|_{\our}\leq \|h_n\|_{\calC^0} \leq \|h_0\|_{\calC^0}\;, \quad \text{ and so } \quad \|h\|_{\our}=\lim_{n \to \infty} \|h_n\|_{\our}\leq \|h_0\|_{\calC^0}\;.$$
Since this holds for arbitrarily large $U$ and $R$, the bound \eqref{hBound} follows.

\vspace{0,2cm}

We will now show that the estimate~\eqref{dhBound} holds. First of all if $h\in\calC^1$ we see that $Dh$ and $\partial_r Dh$  are both continuous, and consequently $D\partial_r h$ exists and its equal to $\partial_r Dh+[D,\partial_r]h$~\footnote{Here we are using the following generalized version of the Schwarz Lemma: if $X$ and $Y$ are two nonvanishing $\calC^1$ vector fields in $\bbR^2$ and $f$ is a $\calC^1$ function such that $X\cdot(Y\cdot f)$ exists and is continuous then $Y\cdot(X\cdot f)$ also exists and is equal to $X\cdot(Y\cdot f) - [X,Y]\cdot f$.}. Using this last fact and equations~\eqref{Definitionh} while differentiating~\eqref{mainEq} with respect to $r$ we obtain an evolution equation for $\partial_rh$:
\begin{equation}
\label{Ddrh}
 D\partial_{r}h=-2\frac{\Lambda}{3}r\,\partial_{r}h\;.
\end{equation}
Integrating the last equation along the (ingoing) characteristics, as before, yields
\newcommand{\meuDL}{\frac{2\Lambda}{3}}
\begin{equation}
\partial_{r}h(u_{1},r_{1})=\partial_{r}h_{0}(r_{0})e^{-\meuDL\int_0^{u_1}r(s)ds}.
\end{equation}
It is then clear that initial data controls the supremum norm of $\partial_rh$. In fact,  let
$$d_0=\|(1+r)^p\partial_rh_0\|_{\calC^0}\;.$$
In the cosmological region ($r>r_c$),
one has, after recalling~\eqref{charCosmic},
\begin{equation}\label{dh1}
\begin{aligned}
\left|(1+r_{1})^{p}e^{Hu_{1}}\partial_{r}h(u_1,r_1)\right|
&=
\left|(1+r_{1})^{p}e^{Hu_{1}}\partial_{r}h_{0}(r_0)e^{-\meuDL\int_0^{u_1}r(s)ds}\right|
\\
&\leq
d_0\left(\frac{1+r_1}{1+r_{0}}\frac{\sinh{\left(\alpha(c-u_1)\right)}}{\sinh{\left(\alpha c\right)}}\right)^{p}e^{Hu_1}\left(\frac{\sinh{\left(\alpha(c-u_1)\right)}}{\sinh{\left(\alpha c\right)}}\right)^{4-p},
\end{aligned}
\end{equation}
where $\alpha=\frac{1}{2}\sqrt{\meuL}$ as before. Now, since $c-u_1\leq c$, then $e^{-2\alpha(c-u_1)}\geq e^{-2\alpha c}$, and
\begin{equation}
\begin{aligned}
\frac{\sinh{\left(\alpha(c-u_1)\right)}}{\sinh{\left(\alpha c\right)}}&=\frac{e^{\alpha(c-u_1)}-e^{-\alpha(c-u_1)}}{e^{\alpha c}-e^{-\alpha c}} \\
                                                                      &=e^{-\alpha u_1}\frac{1-e^{-2\alpha(c-u_1)}}{1-e^{-2\alpha c}} \\
                                                                      &\leq e^{-\alpha u_1}\;.
\end{aligned}
\label{dh2}
\end{equation}
Also
\begin{equation}\label{dh3}
\begin{aligned}
\frac{1+r_1}{1+r_0}\frac{\sinh{\left(\alpha(c-u_1)\right)}}{\sinh{\left(\alpha c\right)}}
&=
\frac{1+\frac{1}{2\alpha}\coth{\left(\alpha(c-u_1)\right)}}{1+\frac{1}{2\alpha}
\coth{\left(\alpha c\right)}}\frac{\sinh{\left(\alpha(c-u_1)\right)}}{\sinh{\left(\alpha c\right)}}
\\
&=
\frac{\sinh\left(\alpha(c-u_1)\right)+\frac{1}{2\alpha}\cosh\left(\alpha(c-u_1)\right)}
{\sinh\left(\alpha(c)\right)+\frac{1}{2\alpha}\cosh\left(\alpha c\right)}
\\
&\leq
\frac{1+\frac{1}{2\alpha}}{\frac{1}{2\alpha}}\cdot\frac{\cosh{\left(\alpha(c-u_1)\right)}}{\cosh{\left(\alpha c\right)}}
\\
&\leq(2\alpha+1)\,2e^{-\alpha u_1}.
\end{aligned}
\end{equation}
Therefore, if $0\leq p\leq 4$ and $H\leq4\alpha=2\sqrt{\Lambda/3}$, we plug~\eqref{dh2} and~\eqref{dh3} into~\eqref{dh1} to obtain
\begin{equation}
\begin{aligned}
\sup_{(u_1,r_1)\in[0,U]\times[r_c,R]}\left|(1+r_1)^{p}e^{Hu_1}\partial_{r}h(u_{1},r_{1})\right|
&\leq
d_{0}\sup_{(u_1,r_1)\in[0,U]\times[r_c,R]}\left|2^{p}(2\alpha+1)^{p}e^{(H-4\alpha)u_1}\right|
\\
&\leq
2^{p}(2\alpha+1)^{p}d_0\;.
\end{aligned}
\end{equation}
Similar, although simpler, computations yield
\begin{equation}
\sup_{(u_1,r_1)\in[0,U]\times[0,r_c]}\left|(1+r_1)^{p}e^{Hu_1}\partial_{r}h(u_{1},r_{1})\right|
\leq 16 \sup_{r_1\in[0,r_c]}\left|(1+r_1)^{p}\partial_{r}h_0(r_1)\right| \leq 16 d_0
\end{equation}
%
%and
%%
%\begin{equation}
%\sup_{u_1\in[0,U]}\left|(1+r_c)^{p}e^{Hu_1}\partial_{r}h(u_{1},r_c)\right|
%\leq \left|(1+r_c)^{p}\partial_{r}h_0(r_c)\right|
%\end{equation}
%%
for the local region. This proves~\eqref{dhBound}.

\vspace{0,2cm}

To finish the proof of Theorem~\ref{mainThm} all is left is to establish the uniform decay statement~\eqref{expDecay}. Start with
\begin{equation*}
\begin{aligned}
\left|h(u,r)-\bar{h}(u,r)\right|
&\leq
\frac{1}{r}\int_0^r\left|h(u,r)-h(u,s)\right|ds
\\
&\leq
\frac{1}{r}\int_0^r\int_s^r\left|\partial_{\rho}h(u,\rho)\right|d\rho\, ds
\\
&\lesssim
\frac{1}{r}\int_0^r\int_s^r\frac{e^{-Hu}}{(1+\rho)^p}d\rho\, ds
\lesssim
\left\{
\begin{array}{lcc}
\frac{e^{-Hu}}{1+r} & , &  2< p\leq 4
\\ \\
r e^{-Hu} & ,  & 0\leq p \leq 2
\end{array}
\right.\;.
\end{aligned}
\end{equation*}
These estimates for $2 < p \leq 4$ are obtained by direct computation; they seem to be the optimal results which follow from this method. The remaining cases, with the exception of $p=0$, are far from optimal. In fact, since we are mainly interested in a qualitative analysis, namely if the decay obtained is or not uniform in $r$ (see Remark~\ref{remMainThm}), the results for $p\leq 2$ were obtained simply using $\frac1{(1+r)^p} \leq 1$.

Using~\eqref{mainEq} we then see that
\begin{equation*}
\begin{aligned}
\left|\partial_{u}h\right|
&=
\left|D h + \frac{1}{2}\left(1-\frac{\Lambda}{3}r^{2}\right)\partial_{r}h\right|
\\
&\leq \left|-\frac{\Lambda}{3}r\left(h-\bar{h}\right)\right|+\frac{1}{2}\left|\left(1-\frac{\Lambda}{3}r^{2}\right)\partial_{r}h\right|
\lesssim
(1+r)^{n(p)} e^{-Hu}\;,
\end{aligned}
\end{equation*}
with $n(p)$ as in the statement of the theorem.

Now since $\partial_{u}h$ is integrable with respect to $u$, by the fundamental theorem of calculus, we see that there exists

$$\lim_{u\rightarrow\infty}h(u,r)=\underline{h}(r)\;.$$
But
$$
\begin{aligned}
|\underline{h}(r_2)-\underline{h}(r_1)|
&=
\lim_{u\rightarrow\infty}|h(u,r_2)-h(u,r_1)|
\\
&\leq
\lim_{u\rightarrow\infty}\left|\int_{r_1}^{r_2}|\partial_{r}h(u,r)|dr\right|
\\
&\lesssim \lim_{u\rightarrow\infty} |r_2-r_1|e^{-Hu}=0\;,
\end{aligned}
$$
and, consequently, there exists $\underline{h}\in\mathbb{R}$ such that
$$\underline{h}(r)\equiv\underline{h}\;.$$
Finally
$$
\begin{aligned}
\left|h(u,r)-\underline{h}\right|
&\leq
\int_u^{\infty}\left|\partial_vh(v,r)\right|dv
\\
&\lesssim
\int_u^{\infty}(1+r)^{n(p)}e^{-Hv}dv\lesssim (1+r)^{n(p)}e^{-Hu}\;.
\end{aligned}
$$

\end{proof}

\begin{Remark}
The same calculation shows that given $R>0$ the solutions of \eqref{mainEqChar} satisfy $|h(u,r)-\underline{h}|\lesssim e^{-Hu}$ uniformly for $r \in [0,R]$, even if $\|(1+r)^p \partial_rh_0\|_{\mathcal{C}^0}$ is not finite.
\end{Remark}

\section{Boundedness and exponential uniform decay for spherical linear waves in de Sitter}

We now translate part of the results in Theorem~\eqref{mainThm} back into results concerning  linear waves in de Sitter.

\begin{thm}

Let $(M,g)$ be de Sitter spacetime with cosmological constant $\Lambda$ and $(u,r,\theta,\varphi)$ Bondi coordinates as in Section~\ref{sectionBondi}.
Let $\phi=\phi(u,r)\in\calC^2\left([0,\infty)\times[0,\infty)\right)$ be a
solution~\footnote{Alternatively one might consider a general solution and infer results about its zeroth spherical harmonic.} to
$$
\square_g\phi=0\;.
$$

Then
\begin{equation}\label{phiBound}
%\sup_{u,r\geq 0}
\left|\phi\right|\leq \sup_{r\geq0}\left|\partial_r\left(r\phi(0,r)\right)\right|\;.
\end{equation}

Moreover, if for some $0\leq p \leq 4$
\begin{equation}\label{condDecay}
\sup_{r\geq 0} \left|(1+r)^p\frac{\partial^2}{\partial r^2}\left(r\phi(0,r)\right)\right|<\infty\;,
\end{equation}
then there exists $\underline{\phi}\in\mathbb{R}$ such that, for $H\leq 2 \sqrt{\frac{\Lambda}{3}}$ ,
\begin{equation}\label{phiDecay}
\left|\phi(u,r)-\underline{\phi}\right|\lesssim (1+r)^{n(p)}e^{-Hu}\;,
\end{equation}
where
\begin{equation}\label{lp}
n(p)=\left\{
\begin{array}{ccc}
0 & , & 2 <p \leq 4 \\
2   & , & 0\leq p \leq 2
\end{array}\right.\;.
\end{equation}
\end{thm}

\begin{proof}
Since $\phi$ is a spherically symmetric $\calC^2$ solution of~\eqref{onda} we saw in Section~\ref{sectionBondi} that $h=\partial_r(r\phi)$ satisfies~\eqref{mainEq}, with $\phi=\bh$. Applying Theorem~\eqref{mainThm} the results easily follow.
\end{proof}

\begin{Remark}
Note that the bound \eqref{phiBound} for $\phi$, unlike the bound \eqref{hBound} for $h$, depends on the derivative of the initial data, so we do have ``loss of derivatives'' in this case.
\end{Remark}

\begin{Remark}
Once more that the powers of $1+r$ obtained are far from optimal, see Remark~\ref{remMainThm}.
\end{Remark}

\begin{Remark}
It should be emphasized that the boundedness and decay results are logically independent. In fact~\eqref{phiBound} follows from~\eqref{hBound}, which in turn is a consequence of a fortunate trick (see proof of Lemma~\ref{calFdecrease}) relying on the non-positivity of the factor of the zeroth order term in~\eqref{mainEq} (here, non-negativity of $\Lambda$) and the fact that $\calF$~\eqref{defIntF} is a contraction in appropriate function spaces; in some sense one is required to prove existence and uniqueness of~\eqref{mainEqChar} in the process. That is no longer the case for obtaining~\eqref{expDecay}, from which uniform decay of $\phi$ follows.
\end{Remark}

\section*{Acknowledgements}

This work was supported by projects PTDC/MAT/108921/2008 and CERN/FP/116377/2010, and by CMAT, Universidade do Minho, and CAMSDG, Instituto Superior T\'ecnico, through FCT plurianual funding. AA thanks the Mathematics Department of Instituto Superior T\'ecnico (Lisbon), where this work was done, for hospitality, and FCT for grant SFRH/BD/48658/2008.

\end{document}